\documentclass[10pt,journal,onecolumn]{IEEEtran}

\usepackage{cite}

\usepackage[cmex10]{amsmath}
\usepackage{array}

\usepackage{fixltx2e}

\usepackage{url}

\usepackage{amssymb}
\usepackage{amsthm}

\newtheorem{thm}{Theorem}[section]
\newtheorem{cor}[thm]{Corollary}
\newtheorem{lem}[thm]{Lemma}

\newtheorem{defn}[thm]{Definition}

\newcommand{\Aut}[1]{\mathrm{Aut}({#1})}
\newcommand{\Iso}{\mathrm{Iso}}

\newcommand{\rect}{\mathrm{Rect}}

\begin{document}

\title{On the Classification of MDS Codes}

\author{Janne~I.~Kokkala,
        Denis~S.~Krotov,
        Patric~R.~J.~{\"O}sterg{\aa}rd
\thanks{
The work of J. I. Kokkala was supported in part by the Aalto ELEC Doctoral School.
The work of D. S. Krotov was supported in part by Grant 13-01-00463-a of the Russian Foundation for Basic Research.
The work of P. R. J. {\"O}sterg{\aa}rd was supported in part by the Academy of Finland under Grant No.\ 132122.
The material in this paper was presented in part at the 4th International Castle Meeting in Coding Theory and Applications, Palmela, Portugal, September 2014.}
\thanks{J. I. Kokkala and P. R. J. {\"O}sterg{\aa}rd are with the Department of Communications and Networking, Aalto University School of Electrical Engineering, 00076 Aalto, Finland}
\thanks{D. S. Krotov is with the Sobolev Institute of Mathematics and the Mechanics
and Mathematics Department, Novosibirsk State University, 630090 Novosibirsk, Russia}
}

\maketitle

\begin{abstract}
A $q$-ary code of length $n$, size $M$, and minimum distance $d$ is called an $(n,M,d)_q$ code. An $(n,q^{k},n-k+1)_q$ code is called a maximum distance separable (MDS) code. In this work, some MDS codes over small alphabets are classified. It is shown that every $(k+d-1,q^k,d)_q$ code with $k\geq 3$, $d \geq 3$, $q \in \{5,7\}$ is equivalent to a linear code with the same parameters. This implies that the $(6,5^4,3)_5$ code and the $(n,7^{n-2},3)_7$ MDS codes for $n\in\{6,7,8\}$ are unique. The classification of one-error-correcting $8$-ary MDS codes is also finished; there are $14$, $8$, $4$, and $4$ equivalence classes of $(n,8^{n-2},3)_8$ codes for $n=6,7,8,9$, respectively. One of the equivalence classes of perfect $(9,8^7,3)_8$ codes corresponds to the Hamming code and the other three are nonlinear codes for which there exists no previously known construction.
\end{abstract}


\section{Introduction}

\IEEEPARstart{A}{\emph{code}} of \emph{length} $n$ over an \emph{alphabet} $\mathcal{A}$ is a subset of $\mathcal{A}^n$. With alphabet size $q = |\mathcal{A}|$, the code is called a \emph{$q$-ary} code. The number of codewords is called the \emph{size} of the code. The \emph{Hamming distance} between two words in $\mathcal{A}^n$ is the number of coordinates in which they differ. The \emph{minimum distance} of a code is the minimum Hamming distance between any two distinct codewords. A code with minimum distance $d$ is able to detect errors in up to $d-1$ coordinates and correct errors in up to $\left\lfloor (d-1)/2\right\rfloor$ coordinates. A $q$-ary code of length $n$, size $M$, and minimum distance $d$ is called an $(n,M,d)_q$ code. 

A code with the alphabet $\mathbb{F}_q$, the finite field of order $q$, is \emph{linear} if the codewords form a vector subspace of $\mathbb{F}_q^n$. For \emph{unrestricted} (that is, either linear or nonlinear) codes, two codes are called \emph{equivalent} if one can be obtained from the other by a permutation of coordinates followed by permutations of symbols at each coordinate separately. We use the notation $C \cong C'$ to denote that codes $C$ and $C'$ are equivalent. Equivalence maintains the Hamming distance between codewords but not linearity.

A general bound for the size of an $(n,M,d)_q$ code is the Singleton bound \cite{S64}, which states that
\[
M \leq q^{n-d+1}.
\]
Codes with $M=q^{n-d+1}$ are called \emph{maximum distance separable (MDS)}.

The Hamming bound, or the sphere-packing bound, states that
\[
M \leq \frac{q^n}{\sum_{i=0}^t\binom{n}{i} (q-1)^i},
\]
where $t=\left\lfloor \frac{d-1}{2}\right\rfloor$ is the number of errors a code with minimum distance $d$ can correct. Codes attaining this bound are called \emph{perfect}. For one-error-correcting codes, $d=3$, and thus
\[
M \leq \frac{q^n}{1 + n (q-1)}.
\]

Even the existence of linear MDS codes with given parameters is in general an open question (see \cite[Chapter~11]{MS77}), and less is known about the unrestricted case. The $(n,q^2,n-1)_q$ codes correspond to sets of mutually orthogonal Latin squares, which have been widely studied \cite{ACD07}. For some results for other unrestricted MDS codes, see \cite{T94,A03,YWJ09,AH14}.

Perfect one-error-correcting MDS codes are $(q+1,q^{q-1},3)_q$ codes. For a prime power $q$, the only linear code up to equivalence with these parameters is the Hamming code, whose parity check matrix contains the maximal number $q+1$ of pairwise linearly independent columns. A natural question is whether codes with the same parameters exist that are not equivalent to linear codes. 

The $(3,2^1,3)_2$ code is trivially unique, and the uniqueness of the $(4,3^2,3)_3$ code is not difficult to prove either. Alderson \cite{A06} showed that the $(5,4^3,3)_4$ code is unique. The nonexistence of Graeco-Latin squares of order $6$ implies the nonexistence of $(7,6^5,3)_6$ codes. The cases $q=5,7,8$ are settled in the present work: the $(6,5^4,3)_5$ and $(8,7^6,3)_7$ codes are unique and there exists four equivalence classes of $(9,8^7,3)_8$ codes.

In the general case, when $q$ is a proper prime power and $q\geq 9$, there exists a $(q+1,q^{q-1},3)_q$ code that is not equivalent to the Hamming code with the same parameters, as demonstrated by an early construction by Lindstr{\"o}m \cite{L69}. Heden \cite{H08} studied certain perfect codes when $q$ is a prime and showed that they are equivalent to linear codes. There exist also constructions for nonlinear perfect codes using more restrictive notions of equivalence, such as \cite{PRV05}.

Shortening the perfect codes gives one-error-correcting $(n,q^{n-2},3)_q$ MDS codes for $3\leq n<q+1$. Our work relies on known classification results of $(n,q^{n-2},3)_q$ MDS codes for $n=4,5$. For $n=4$, the codes are equivalent to Graeco-Latin squares of order $q$, which have been classified for $q \leq 8$ by McKay \cite{M}; there are $1$, $1$, $1$, $0$, $7$, $2165$ equivalence classes of such codes for $q=3,4,\dots,8$, respectively. For $n=5$, the codes are equivalent to Graeco-Latin cubes which have been classified recently \cite{KO14}; there are $1$, $1$, and $12484$ equivalence classes of such codes for $q=5,7,8$, respectively.

This work consists of two parts. In the first part, we show that every $(k+d-1,q^k,d)_q$ code, where $k,d\geq 3$ and $q=5,7$, is equivalent to a linear code. For one-error-correcting codes, this implies that the $(6,5^4,3)_5$ code and the $(n,7^{n-2},3)_q$ codes for $n=6,7,8$ are unique. This part is easier to carry out using the terminology of Latin squares. In the second part, we present an algorithm for exhaustive generation of $(n,q^{n-2},3)_q$ codes starting from $(n-1,q^{n-3},3)_q$ codes. Running this algorithm for $q=8$ yielded $14$, $8$, $4$, and $4$ equivalence classes of $(n,8^{n-2},3)_8$ codes for $n=6,7,8,9$, respectively.

\section{Preliminaries}

For ease of notation, we denote 
$[m] = \{1,2,\dots,m\}$
when referring to sets of indices.

\subsection{Latin Hypercubes and MDS Codes}

A \emph{Latin square} of order $q$ is a $q\times q$ array of symbols from an alphabet $\mathcal{A}$ of size $q$ such that each symbol appears exactly once in each row and each column. Two Latin squares are called \emph{orthogonal} if each pair of symbols occurs exactly once when the squares are superimposed. A pair of orthogonal Latin squares is called a \emph{Graeco-Latin square}.

A \emph{Latin hypercube} of dimension $k$ is a $q \times q \times \cdots \times q$ ($k$ times) array of symbols from an alphabet $\mathcal{A}$ of size $q$ where each $q \times q$ subarray, obtained by fixing any $k-2$ coordinates, is a Latin square. Two Latin hypercubes of same dimension are called orthogonal if when the hypercubes are superimposed, every $q\times q$ subarray is a Graeco-Latin square. A pair of Latin hypercubes is called a \emph{Graeco-Latin hypercube}.

We denote the positions in a Latin hypercube of dimension $k$ by elements in $\mathcal{A}^k$, so Latin hypercubes can be viewed as functions from $\mathcal{A}^k$ to $\mathcal{A}$. For ease of notation, we assume that $\mathcal{A} = \mathbb{F}_q$ when $q$ is a prime power unless otherwise mentioned.

There is a one-to-one correspondence between Latin hypercubes of order $q$ and dimension $k$ and $(k+1,q^k,2)_q$ codes: let $c=(c_1,c_2,\dots,c_{k+1})$ be a codeword if $c_{k+1}$ occurs at position $(c_1,c_2,\dots,c_k)$ in the Latin hypercube. Similarly, there is a one-to-one correspondence between Graeco-Latin hypercubes of order $q$ and dimension $k$ and $(k+2,q^k,3)_q$ MDS codes: let $c=(c_1,c_2,\dots,c_{k+2})$ be a codeword if $(c_{k+1},c_{k+2})$ occurs at position $(c_1,c_2,\dots,c_k)$ in the Graeco-Latin hypercube.

We define linearity of Latin hypercubes and tuples of Latin hypercubes as follows. A Latin hypercube $f$ of order $q$ and dimension $k$ is linear if there are permutations $\alpha_0,\alpha_1,\dots,\alpha_k$ of $\mathbb{F}_q$ such that 
\begin{equation}
\alpha_0(f(x_1,x_2,\dots,x_k))  
 = \alpha_1(x_1) + \alpha_2(x_2) + \cdots + \alpha_k(x_k). \label{eq:def-lin}
\end{equation}
This is equivalent to the condition that the corresponding MDS code be equivalent to a linear code. An $r$\nobreakdash-tuple of (not necessarily mutually orthogonal) Latin hypercubes $(f_1,f_2,\dots,f_r)$ is linear if there are permutations $\alpha_1,\alpha_2,\dots,\alpha_k,\beta_1,\beta_2,\dots,\beta_r$ of $\mathbb{F}_q$ and coefficients $a_{i,j} \in \mathbb{F}_q$ for $i\in[r]$, $j \in [k]$ such that
\[
\beta_i(f_i(x_1,x_2,\dots,x_k)) 
= a_{i,1}\alpha_1(x_1) + a_{i,2}\alpha_2(x_2) +\cdots + a_{i,k} \alpha_k(x_k),
\]
for each $i \in [r]$.
We may assume that $a_{1,i}=1$ for all $i$.
For Graeco-Latin hypercubes, this is equivalent to the condition that the corresponding MDS code be equivalent to a linear code.

\subsection{Properties of MDS Codes}

Codes can be transformed into shorter and longer codes by operations called shortening and extending. Because these operations are used extensively in the description of the algorithm, we introduce precise notation for them here.

\begin{defn}
For an $(n,M,d)_q$ MDS code $C$, let 
\[
s(C,i,v) 
 = \{ (c_1,c_2,\dots,c_{i-1},c_{i+1},\dots,c_n) : c \in C \text{ and } c_i = v \}.
\]
This operation is called \emph{shortening}.
\end{defn}

\begin{defn}
For an $(n,M,d)_q$ MDS code $C$, let 
\[
e(C,i,v) = \{ (c_1,c_2,\dots,c_{i-1},v,c_{i},\dots,c_n) : c \in C\}.
\]
This operation is called \emph{extending}.
\end{defn}

In other words, $s(C,i,v)$ is the $(n-1,M',d')_q$ code that is obtained by removing the $i$th coordinate from $C$ and retaining the codewords that have $v$ at that coordinate, and $e(C,i,v)$ is the $(n+1,M,d)_q$ code which is obtained by adding a coordinate at $i$ with the symbol $v$ to each codeword of $C$. 

The following basic theorems are important in the construction of MDS codes based on shorter codes presented in Section~\ref{sec:computational}.

\begin{thm}
A shortened MDS code is an MDS code.
\end{thm}
\begin{thm}
An $(n,q^{k},n-k+1)_q$ MDS code is a union of $q$ extended MDS codes: for each coordinate $i$ there are $(n-1,q^{k-1},n-k+1)_q$ MDS codes $C'_v$ for each $v\in \mathbb{F}_q$ such that
\[
C = \bigcup_{v \in \mathbb{F}_q} e(C'_v,i,v).
\]
\end{thm}
\begin{proof}
Simply choose $C'_v = s(C,i,v)$.
\end{proof}

\subsection{Code Equivalence} \label{sec:equivalence}

The operations maintaining equivalence of codes of length $n$ and alphabet $\mathcal{A}$ form a group $G$ that acts on $\mathcal{A}^n$. Each element $g\in G$ can be expressed in terms of a permutation $\pi$ of $[n]$ and permutations $\sigma_i$ for $i\in [n]$  of $\mathcal{A}$ as
\[
g = (\pi; \sigma_1, \sigma_2, \dots, \sigma_n),
\]
such that for each $c = (c_1,c_2,\dots,c_n) \in \mathcal{A}^n$ and for each $i\in [n]$,
\[
(gc)_{\pi(i)} = \sigma_{\pi(i)}(c_i),
\]
where $(gc)_i$ denotes the $i$th symbol of $gc$. 

Two codes, $C$ and $C'$, are thus equivalent when there exists a $g\in G$ such that $C=gC'$. The set of all elements of $G$ that map $C$ to $C'$ is denoted by $\Iso(C,C')$. An element of $\Iso(C,C)$ is called an automorphism of $C$. The group of automorphisms of $C$ is denoted by $\Aut{C}$.
For equivalent codes $C$ and $C'$, we can write
\begin{equation}
\Iso(C,C') = \Aut{C'} g, \label{eq:isoaut}
\end{equation}
where $g$ is any element of $\Iso(C,C')$.

Each word that has value $v$ at coordinate $i$ is mapped by $g$ to a word that has value $\sigma_{\pi(i)}(v)$ at coordinate $\pi(i)$. We also define an action of $G$ on $[n] \times \mathcal{A}$ by
\[
g(i,v) = (\pi(i),\sigma_{\pi(i)}(v)).
\]

When the length of the codes is not obvious from the context, we denote by $G=G_n$ the group acting on $\mathcal{A}^n$. Because the study of equivalence of shortened codes of two codes plays a crucial role in the algorithm, we need the following two definitions to ease notation.

\begin{defn}
For every $g \in G_n$ and every $i \in [n]$, define $e(g,i) \in G_{n+1}$ to be the element that applies $g$ to the subcodes obtained by removing $i$ and keeps the coordinate $i$ intact, that is,
\[
  e(gC,i,v) = e(g,i)e(C,i,v),
\]
for every $v\in\mathcal{A}$, and $C \subseteq \mathcal{A}^n$.
\end{defn}

\begin{defn}
For every $g \in G_{n}$ and every $i \in [n]$ such that $g$ maps coordinate $i$ to itself and does not permute the symbols in coordinate $i$, define $s(g,i) \in G_{n-1}$ such that it applies $g$ ignoring the coordinate $i$ to codes of length $n-1$, that is,
\[
  s(gC,i,v) = s(g,i)s(C,i,v),
\]
for each $v\in\mathcal{A}$ and $C \subseteq \mathcal{A}^n$.
\end{defn}

\subsection{Computational Tools}

To solve the problem of code equivalence computationally, we reduce it to the graph isomorphism problem. For each $q$\nobreakdash-ary code $C$ of length $n$, we define a labeled coloured graph as follows. The graph contains $n$ copies of the complete graph with $q$ vertices, colored with the first colour. For each codeword, the graph contains a vertex colored with the second color. From a vertex corresponding to codeword $c$, there is an edge to the $v$th vertex in the $i$th complete graph if and only if $c$ has a value $v$ at coordinate $i$.

Now two codes, $C$ and $C'$, are equivalent if and only if their corresponding graphs, $H$ and $H'$, respectively, are isomorphic. The permutation of coordinates corresponds to permutation of the complete graphs, and the permutations of symbols in each coordinate corresponds to permutation of vertices in each complete graph. Moreover, in a graph isomorphism mapping $H$ to $H'$, the permutation of the vertices of the first colour uniquely determines the permutation of the vertices of the second colour, so there is a direct correspondence between $\Iso(C,C')$ and the set of graph isomorphisms from $H$ to $H'$.

The software \emph{nauty} \cite{MP14} can be used to find canonical labelings of graphs, which then can be used to find a graph isomorphism between isomorphic graphs. In addition, \emph{nauty} returns the automorphism of a graph. Along with \eqref{eq:isoaut}, this allows finding the set $\Iso(C,C')$ for two codes $C$ and $C'$. We use \emph{nauty} in the sparse mode with the random Schreier method enabled.

\section{Theoretical Results}

In this section, we show that an $r$-tuple of Latin hypercubes of prime order and dimension $k$, where $r\geq 2$ and $k\geq 3$, is linear if each pair of Latin hypercubes of dimension $3$ obtained by fixing $k-3$ coordinates from two hypercubes of the tuple is linear. We start by showing that every Latin hypercube of prime order and dimension $k$, where $k\geq 4$, is linear if every Latin hypercube obtained from it by fixing one coordinate is linear.

\begin{defn} A \emph{rectangle} of directions $i$ and $j$ ($i \neq j$) is a quadruple $(a=(a_1,a_2,\dots,a_k),b=(b_1,b_2,\dots,b_k),c=(c_1,c_2,\dots,c_k),d=(d_1,d_2,\dots,d_k))$ of elements of $\mathbb{F}_q^k$ such that $a_i=b_i$, $c_i=d_i$, $b_j = c_j$, and $d_j=a_j$ and $a_l=b_l=c_l=d_l$ for all $l \in [k]\setminus \{i,j\}$.
\end{defn}

\begin{lem}
For every linear Latin hypercube $f$ of prime order $q$ there is a unique function $\rect_f : \mathbb{F}_q^3 \to \mathbb{F}_q$ such that for every rectangle $(a,b,c,d)$,
\[
f(a) = \rect_f(f(b),f(c),f(d)).
\]
\end{lem}
\begin{proof}
Using the notation in \eqref{eq:def-lin}, we find that 
\[
f(a) = \alpha_0^{-1}(\alpha_0(f(b)) - \alpha_0(f(c)) + \alpha_0(f(d))).
\]
\end{proof}

\begin{lem} \label{lem:uniqrecon}
A linear Latin hypercube $f$ of order $q$ can be uniquely reconstructed from the function $\rect_f$ and the values $f(x_1,x_2,\dots,x_k)$ where at most one of $x_i$ is nonzero.
\end{lem}
\begin{proof}
When $x$ has $m\geq 2$ nonzero elements, the value $f(x)$ can be uniquely determined from the function $\rect_f$ and the values $f(x')$ where $x'$ has $m-1$ nonzero elements using
\begin{multline*} 
f(x_1,\dots,x_i,\dots,x_j,\dots,x_k) 
 \\ 
= \rect_f(
f(x_1,\dots,0,\dots,x_j,\dots,x_k), 
 f(x_1,\dots,0,\dots,0,\dots,x_k),
 f(x_1,\dots,x_i,\dots,0,\dots,x_k)).
\end{multline*}
The lemma follows by induction on $m$.
\end{proof}

\begin{lem} \label{lem:hyperlin}
Let $f$ be a hypercube of dimension $k$, where $k\geq 4$, such that each $(k-1)$-dimensional Latin hypercube obtained from $f$ by fixing one argument is linear. Then $f$ is linear.
\end{lem}
\begin{proof}
For $j \in [k]$, let $r_j$ be the linear hypercube of dimension $k-1$ obtained from $f$ by letting the $j$th argument be $0$.
Without loss of generality, we may assume that 
\[
r_n(x_1,x_2,\dots,x_{k-1}) = x_1 + x_2 + \cdots + x_{k-1},
\]
and that 
\[
f(0,0,\dots,0,x_k) = x_k.
\]
Now
\[
\rect_{r_k} (a,b,c) = a-b+c.
\]
For $j\in [k-1]$, let $s_j$ be the linear hypercube of dimension $k-2$ obtained by letting the $j$th and the $k$th argument of $f$ be $0$. Because $s_j$ occurs as a subarray in both $r_j$ and $r_k$, we get
\[
\rect_{r_j} = \rect_{s_j} = \rect_{r_k}.
\]
Because 
\[
r_j(0,0,\dots,0,x_i,0,\dots,0) = x_i,
\]
where $i \in [k-1]$ and $x_i$ occurs in the $i$th position, Lemma~\ref{lem:uniqrecon} implies that 
\[
r_j(x_1,x_2,\dots,x_{k-1}) = x_1 + x_2 + \cdots + x_{k-1},
\]
for each $j \in [k]$, or equivalently,
\begin{equation}
f(x_1,x_2,\dots,x_k) = x_1+x_2+\cdots + x_k, \label{eq:hcubezero}
\end{equation}
when $x_i = 0$ for at least one value of $i$.

For each $a \in \mathbb{F}_q$, let $t_a$ be the Latin hypercube obtained from $f$ by letting the last argument be $a$. Now
\[
t_a(0,0,\dots,0,x_i,0,\dots,0) = x_i + a,
\]
where $x_i$ occurs in the $i$th position. The function $\rect_{t_a}$ is determined by \eqref{eq:hcubezero}, and again by Lemma~\ref{lem:uniqrecon}, we get that 
\[
t_a(x_1,x_2,\dots,x_{k-1}) = x_1 + x_2 + \cdots + x_{k-1} + a,
\]
for all $a$, or equivalently
\[
f(x_1,x_2,\dots,x_k) = x_1 + x_2 + \cdots + x_k.
\]
Thus, $f$ is linear.
\end{proof}

We need one more lemma before proving the main theorem.
\begin{lem} \label{lem:affine}
Let $q$ be a prime, let $c \in \mathbb{F}_q$, let $a_1,a_2,a_3 \in \mathbb{F}_q \setminus \{0\}$, and let $\gamma_1$, $\gamma_2$, and $\gamma_3$ be permutations of\/ $\mathbb{F}_q$. If 
\[
\gamma_1(x_1) + \gamma_2(x_2) + \gamma_3(x_3) = c
\]
 whenever
\[
a_1 x_1 + a_2 x_2 + a_3 x_3 = 0,
\]
then $\gamma_i$ is an affine transformation of $\mathbb{F}_q$, for all $i$.
\end{lem}
\begin{proof}
For all $x \in \mathbb{F}_q$, we find that
\[
\gamma_1(x+1 ) - \gamma_1(x) 
=\left[c - \gamma_2(-a_2^{-1}a_1 x) - \gamma_3(-a_3^{-1} a_1) \right]  
-\left[c - \gamma_2(-a_2^{-1}a_1 x) - \gamma_3(0) \right]
=\gamma_3(0) - \gamma_3(-a_2^{-1} a_1).
\]
Because $1$ generates the additive group of $\mathbb{F}_q$, we get
\[
\gamma_1(x) = [\gamma_3(0) - \gamma_3(-a_3^{-1} a_1)] x + \gamma_1(0),
\]
for each $x \in \mathbb{F}_q$. Thus, $\gamma_1$ is an affine transformation. By symmetry, so are $\gamma_2$ and $\gamma_3$.
\end{proof}

\begin{thm} \label{thm:linearlatin}
Let $(f_1,f_2,\dots,f_r)$ be an $r$-tuple of Latin hypercubes of prime order $q$ and dimension $k$, with $r\geq 2$ and $k\geq 4$, such that each pair of Latin cubes obtained from any pair of them by fixing the same $k-3$ arguments is linear. Then $(f_1,f_2,\dots,f_r)$ is a linear $r$-tuple of Latin hypercubes. 
\end{thm}
\begin{proof}
By induction and Lemma~\ref{lem:hyperlin}, $f_i$ is a linear Latin hypercube for each $i$. Without loss of generality, we may assume that
\[
f_i(x_1,x_2,\dots,x_k) 
= \gamma_{i,1}(x_1) + \gamma_{i,2}(x_2) + \cdots + \gamma_{i,k}(x_k),
\]
for each $i\in [r]$, where $\gamma_{i,j}$ are permutations of $\mathbb{F}_q$ and $\gamma_{1,j}$ is the identity for each $j\in[k]$.

Consider some $i\in[r]$ and distinct $j_1,j_2,j_3 \in [k]$. Letting all arguments except $j_1$, $j_2$, $j_3$ of $f_1$ and $f_i$ be $0$, we obtain a linear pair $(g,h)$ of Latin hypercubes of dimension $3$ for which
\begin{IEEEeqnarray*}{rCl} 
\beta_0(g(x_1,x_2,x_3)) &=& \beta_0(x_1 + x_2 + x_3) 
=\alpha_1(x_1) + \alpha_2(x_2) + \alpha_3(x_3), \\
\beta_1(h(x_1,x_2,x_3)) &=& \beta_1(\gamma_{i,j_1}(x_1) + \gamma_{i,j_2}(x_2) + \gamma_{i,j_3}(x_3)) 
=a_1\alpha_1(x_1) + a_2\alpha_2(x_2) + a_3\alpha_3(x_3),
\end{IEEEeqnarray*}
for some $a_1,a_2,a_3\in\mathbb{F}_q$ and permutations $\beta_0,\beta_1,\alpha_1,\alpha_2,\alpha_3$ of $\mathbb{F}_q$.

Because $\beta_0(g(x_1,x_2,x_3))=\beta_0(0)$ whenever $x_1+x_2+x_3 = 0$, we see by Lemma~\ref{lem:affine} that $\alpha_1$, $\alpha_2$, and $\alpha_3$ are affine transformations. Similarly, $h(x_1,x_2,x_3)$ is a function of $b_1x_1+b_2x_2+b_3x_3$ for some $b_1,b_2,b_3 \in \mathbb{F}_q$, and thus $\gamma_{i,j_l}$ is an affine transformation for each $l\in\{1,2,3\}$. 

Therefore, $\gamma_{i,j}$ is an affine transformation for all $i\in[r]$ and $j\in[k]$. Thus, $(f_1,f_2,\dots,f_r)$ is a linear $r$-tuple of Latin hypercubes.
\end{proof}

Using the known computational results for Graeco-Latin cubes of orders $5$ and $7$, Theorem~\ref{thm:linearlatin} implies the following.

\begin{thm} \label{thm:linearcode}
Every code with parameters $(k+d-1,7^k,d)_7$ or $(k+d-1,5^k,d)_5$, where $k,d\geq 3$, is equivalent to a linear code.
\end{thm}
\begin{proof}
For every $(n,q^k,d)_q$ code $C$ with $n=k+d-1$, there is a $(d-1)$-tuple of mutually orthogonal Latin hypercubes $(f_1,f_2,\dots,f_{d-1})$ of order $q$ and dimension $n$ such that $C$ is the set of $n$-tuples $(x_1,x_2,\dots,x_n)$ that satisfy
\begin{IEEEeqnarray*}{rCl}
f_1(x_1,x_2,\dots,x_k) &=& x_{k+1}, \\
f_2(x_1,x_2,\dots,x_k) &=& x_{k+2}, \\
&\vdots& \\
f_{d-1}(x_1,x_2,\dots,x_k) &=& x_{k+d-1}.
\end{IEEEeqnarray*}
Because every Graeco-Latin cube of order $5$ or $7$ is linear, $(f_1,\dots,f_k)$ is a linear $(d-1)$-tuple of Latin hypercubes for $q=5,7$ by Theorem~\ref{thm:linearlatin}. Therefore, $C$ is equivalent to a linear code.
\end{proof}

\begin{cor}[MDS conjecture for $q=5,7$]
For $q\in \{5,7\}$, $k\geq 2$ and $d=n-k+1 > 2$, there exists an $(n,q^{k}, n-k+1)_q$ MDS code if and only if $n\leq q+1$.
\end{cor}
\begin{proof}
The case $k=2$ follows from the well known theorem that the size of a set of mutually orthogonal Latin squares of order $q$ is at most $q-1$. We have shown that the existence of any MDS code for $k\geq 3$, $d \geq 3$ implies the existence of a linear code with the same parameters, and the MDS conjecture is true for linear codes over prime fields \cite{B12}.
\end{proof}

\begin{lem} \label{lem:linearcode}
Let $q$ be a prime power and $n \in \{q-1,q,q+1\}$. All linear $(n,q^{n-2},3)_q$ codes are equivalent.
\end{lem} 
\begin{proof}
Let $\alpha$ be a primitive element of $\mathbb{F}_q$. After multiplying each column by a scalar, the parity check matrix of an $(n,q^{n-2},3)_q$ code can be written as
\[
\begin{pmatrix}
0 & 1 & 1 & \cdots & 1 \\
1 & a_1 & a_2 & \cdots & a_{n-1}
\end{pmatrix},
\]
where all $a_i$ are distinct.
Because at most two elements from $\mathbb{F}_q$ are missing from $S=\{a_1,a_2,\dots,a_{n-1}\}$ when $n\geq q-1$, there is an affine transformation $x\mapsto bx + c$ with $b\neq 0$ that maps $S$ to $\{0,1,\alpha^1,\alpha^2, \dots,\alpha^{n-2}\}$. Multiplying the second row by $b$, adding the first row multiplied by $c$ to the first row, multiplying the first column by $b^{-1}$
and permuting the columns yields
\[
\begin{pmatrix}
0 & 1 & 1 & 1 & \cdots & 1 \\
1 & 0 & 1 & \alpha^1 & \cdots & \alpha^{n-2}
\end{pmatrix}.
\]
Because elementary row operations on the parity check matrix do not change the code and multiplying a column and permuting columns maintain equivalence, every linear $(n,q^{n-2},3)_q$ code is equivalent to the code with the parity check matrix described above.
\end{proof}
\begin{cor}
The $(6,5^4,3)_5$ code and the $(n,7^{n-2},3)_7$ codes for $n=6,7,8$ are unique.
\end{cor}
\begin{proof}
By Theorem~\ref{thm:linearcode} these codes are linear, and by Lemma~\ref{lem:linearcode} they are equivalent.\end{proof}

\section{Computational Classification} \label{sec:computational}

\subsection{Algorithm}

The algorithm to be presented generates representatives of all equivalence classes of $(n+1,q^{n-1},3)_q$ codes using an ordered set of representatives of equivalence classes of $(n,q^{n-2},3)_q$ codes, denoted by $\hat{S}^{n} = \{ \hat{C}^n_1, \hat{C}^n_2, \dots, \hat{C}^n_N\}$. For simplicity, we assume that every $\hat{C}^{n}_k$ contains the all-zero codeword.
\begin{defn}
Let $\phi$ be a function that maps each $(n,q^{n-2},3)_q$ code $C$ to an integer in $[N]$ such that $C \cong \hat{C}^{n}_{\phi(C)}$.
\end{defn}

To reduce the search tree and the number of equivalent codes generated, we construct only $(n+1,q^{n-1},3)_q$ codes and their subsets of a certain form. More precisely, we call a subset $C$ of $\mathbb{F}_q^{n+1}$ \emph{semi-canonical} if it satisfies the following properties:

\begin{enumerate}
\item $C$ has minimum distance $3$,
\item\label{itm:A} $s(C,1,0) = \hat{C}^{n}_k$ for some $k$, 
\item\label{itm:B} For all $i\in[n+1]$ and $v \in \mathbb{F}_q$ for which $s(C,i,v)$ has $q^{n-2}$ codewords, $\phi(s(C,i,v)) \geq k$.
\end{enumerate}
Every $(n+1,q^{n-1},3)_q$ code $C$ is equivalent to a code that satisfies these properties.

The central part of the algorithm is a procedure which, given an index $k$, a coordinate $i \in [n]$, and $v \in \mathbb{F}_q$, finds, up to a permutation of the values $\mathbb{F}_q\setminus\{0\}$ in the first coordinate, all possible $(n,q^{n-2},3)_q$ codes $C$ for which 
\[
e(\hat{C}^{n}_k,1,0) \cup e(C,i+1,v)
\]
is semi-canonical. A necessary condition is that 
\begin{equation}
s(C,1,0) = s(\hat{C}^{n}_k,i,v). \label{eq:algonecess}
\end{equation}
The following theorem yields a way to exhaustively construct the codes $C$ satisfying the above condition.

\begin{defn}
For each $i \in [n]$ and $v \in \mathbb{F}_q$, let $h_{i,v} \in G_{n}$ be the element that applies the cyclic permutation $(1\,2\,\cdots\,i)$ to the coordinates and then swaps the values $v$ and $0$ in the first coordinate.
\end{defn}

\begin{thm} \label{thm:algorithm}
Let $\hat{C}$ be an $(n,q^{n-2},3)_q$ code and let $\hat{D}$ be an $(n-1,q^{n-3},3)_q$ code. Let $C$ be a code equivalent to $\hat{C}$ for which $s(C,1,0) = \hat{D}$. Now $C$ can be expressed as
\[
C = g' e(g,1) h_{i,v} \hat{C},
\]
where  $g' \in G_n$ permutes the values $\mathbb{F}_q\setminus\{0\}$ in the first coordinate and keeps other coordinates intact, $(i,v)$ is a coordinate-value pair, and $g \in \Iso(s(h_{i,v} \hat{C},1,0), \hat{D})$.
\end{thm}
\begin{proof}
Let $g'' \in G_n$ such that $C=g''\hat{C}$. Let $(i,v) = g''^{-1}(1,0)$. Because
$g'' h_{i,v}^{-1} (1,0) = (1,0)$,
we can express $g'' h_{i,v}^{-1}$ as 
\[
g'' h_{i,v}^{-1} = g' e(g,1,0),
\]
where $g'$ permutes the nonzero values in the first coordinate and keeps other coordinates intact and $g = s(g''h_{i,v}^{-1},1) \in G_{n-1}$. We obtain
\[
\hat{D} = s(C,1,0) = s(g''\hat{C},1,0)
=s(g' e(g,1,0) h_{i,v} \hat{C}, 1, 0) = g s(h_{i,v} \hat{C},1,0),
\]
and thus $g \in \Iso(s(h_{i,v} \hat{C},1,0), \hat{D})$.
\end{proof}

The codes $C$ satisfying \eqref{eq:algonecess} are now generated with the following algorithm. We loop over all $l=k,k+1,\dots,|\hat{S}^{n}|$ and all coordinate-value pairs $(j,w)$ for which $s(\hat{C}^{n}_l,j,w) \cong s(\hat{C}^{n}_k,i,v)$. In each step, we loop over all $g \in \Iso(s(h_{j,w}\hat{C}^{n}_l,1,0), s(\hat{C}^{n}_k,1,0))$ and consider the code 
\begin{equation}
C = e(g,1)h_{j,w}\hat{C}^{n}_l, \label{eq:code-subprocedure}
\end{equation}
and report it if 
\[
e(\hat{C}^{n}_k,1,0) \cup e(C,i+1,v)
\]
has minimum distance $3$.

We generate the $(n+1,q^{n-1},3)_q$ codes in two phases. In the first phase, we consider codes containing the codewords that have a $0$ in the first or the second coordinate. These codes are potential subsets of $(n+1,q^{n-1},3)_q$ codes. More precisely, we construct, for each $k$ separately, the semi-canonical codes that are of the form
\[
e(\hat{C}^{n}_k,1,0) \cup e(C,2,0),
\]
where $C$ has the property that for all $v\in\mathbb{F}_q$ there is a $w\in\mathbb{F}_q$ such that $C$ contains the codeword $v00..0vw$. These codes form the seeds for the next phase. The permutation of the nonzero values in the first coordinate of $C$ can be chosen to satisfy the last requirement, so the seeds can be constructed by the procedure described above. We perform isomorph rejection on the obtained seeds, since equivalent seeds would be augmented to equivalent codes.

In the second phase, we start from a seed 
\[
C = e(\hat{C}^{n}_k,1,0) \cup e(C',2,0)
\]
and find all semi-canonical $(n+1,q^{n-1},3)_q$ codes that have $C$ as a subset. These codes can be written in the form
\[
\bigcup_{v \in \mathbb{F}_q} e(C''_v,3,v),
\]
where each $C''_v$ is an $(n,q^{n-2},3)_q$ code with the following properties:
\begin{itemize}
\item $\phi(C''_v) \geq k$,
\item $e(\hat{C}^{n}_k,1,0) \cup e(C''_v,3,v)$ has minimum distance $3$,
\item $e(C',2,0) \cup e(C''_v,3,v)$ has minimum distance $3$.
\end{itemize}
The first two properties allow us to find all possible choices for the code $C''_v$ using the procedure described above. The third property implies 
\[
s(C',2,v) = s(C''_v,2,0),
\]
which either rejects a code immediately or yields a unique permutation of the values in the first coordinate of $C''_v$. The requirement that $e(C',2,0) \cup e(C''_v,3,v)$ have minimum distance $3$ can also be used to reject some choices. When all possible choices for $C''_v$ for each $v$ have been generated, we loop over all sets of $C''_v$ for $v \in \mathbb{F}_q$ and report
\[
D = \bigcup_{v \in \mathbb{F}_q} e(C''_v,3,v)
\]
if it is semi-canonical.

Most time is spent using \emph{nauty} to detect code equivalence, so an obvious way to optimize performance is to reduce the number of code equivalence instances that need to be solved. For example, detecting the equivalence class where each shortened code $s(\hat{C}^n_k,i,v)$ belongs needs to be done only when generating the codes of length $n$, and the results can be used when generating the codes of length $n+1$.
In addition, when generating codes in \eqref{eq:code-subprocedure}, we can consider only one $(j,w)$ from each orbit of the coordinate-value pairs in the automorphism group of $\hat{C}^n_l$. 

\subsection{Results}

The algorithm was run for the case $q=8$ starting from the representatives of the $12484$ equivalence classes of $(5,8^3,3)_8$ codes constructed in \cite{KO14} and proceeding step by step to the $(9,8^7,3)_8$ codes. The search yielded $14$, $8$, $4$, and $4$ equivalence classes of $(n,8^{n-2},3)_8$ codes for $n=6,7,8,9$, respectively. The orders of the automorphism groups of the codes are given in Table~\ref{tab:aut}. One of the equivalence classes of perfect codes correspond to the Hamming code, and the other three are new nonlinear codes for which no known construction exists; for example, the construction in \cite{PRV05} is equivalent to the linear code with the present definition of code equivalence. The nonlinear codes are presented in the Appendix.

\begin{table}
\renewcommand{\arraystretch}{1.1}
\centering
\caption{Automorphism Group Orders of $(n,8^{n-2},3)_8$ Codes}
\label{tab:aut}
\begin{tabular}{rr|rr}
\hline
\multicolumn{2}{r|}{$n=6$} & \multicolumn{2}{r}{$n=7$} \\
$|\Aut{C}|$ & $\#$ & $|\Aut{C}|$ & $\#$  \\
\hline
1\,536 & 3 & 16\,384 & 1 \\
2\,048 & 1 & 24\,576 & 1 \\
3\,072 & 1 & 65\,536 & 2 \\
4\,096 & 5 & 86\,016 & 1 \\
12\,288 & 3 & 98\,304 & 1 \\
516\,096 & 1 & 196\,608 & 1 \\
& & 9\,633\,792 & 1 \\
\hline
\multicolumn{2}{r|}{$n=8$} & \multicolumn{2}{r}{$n=9$} \\
$|\Aut{C}|$ & $\#$ & $|\Aut{C}|$ & $\#$  \\
\hline
393\,216 & 1 & 25\,165\,824 & 1 \\
688\,128 & 1 & 44\,040\,192 & 1 \\
786\,432 & 1 & 50\,331\,648 & 1 \\
308\,281\,344 & 1 & 22\,196\,256\,768 & 1 \\
\hline
\end{tabular}
\end{table}

We give in Table~\ref{tab:res}, for each $n$ separately, the number of seeds before and after isomorph rejection and the number of codes the inequivalent seeds were augmented to, again before and after isomorph rejection. The time required for the search for each $n$ is also given and corresponds to one core of an Intel Xeon E5-2665 processor. The time for case $n$ includes the search for seeds and augmenting seeds, isomorph rejection after both steps, and identifying the shortened codes of obtained $(n,q^{n-2},3)_q$ codes to detect whether the codes are semi-canonical. These results can also be used when generating $(n+1,q^{n-1},3)_q$ codes, so the time requirement of a step would be higher if no previous results were available.

\begin{table*}
\centering
\caption{Details of the Search}
\label{tab:res}
\begin{tabular}{r|rrr|r|r}
\hline
$n$ & \# of seeds & \# of inequivalent seeds & \# of codes & \# of inequivalent codes & CPU time (hours) \\
\hline
$6$ & 122 & 107 & 21 & 14 &   15 \\
$7$ &  15 &   9 &  9 &  8 &   49 \\
$8$ &   9 &   6 &  6 &  4 &  340 \\
$9$ &   4 &   4 &  4 &  4 & 1516 \\
\hline
\end{tabular}
\end{table*}

\subsection{Consistency Check}

To check the consistency of the results given by the algorithm, we count for each $k$ in two ways the number $N_k$ of semi-canonical $(n+1,q^{n-1},3)_q$ codes $C$ for which $s(C,1,0)=\hat{C}^{n}_k$.

The first count is obtained by detecting subcodes of the $(n+1,q^{n-1},3)_q$ codes codes obtained. For an $(n+1,q^{n-1},3)_q$ code $C$ and an $(n,q^{n-2},3)_q$ code $C'$, let $S(C,C')$ be the number of pairs $(i,v)$ such that $s(C,i,v) \cong C'$. Let $\mathcal{S}_k$ be the set of obtained inequivalent $(n+1,q^{n-1},3)_q$ codes $C$ for which $\min_{i,v} \phi(s(C,i,v)) = k$. Consider an arbitrary $C \in \mathcal{S}_k$. The size of the equivalence class of $C$ is simply $|G_{n+1}| / |\Aut{C}|$. The proportion of the codes $C'$ in the equivalence class for which $s(C',1,0) \cong \hat{C}^{n}_k$ is $S(C,\hat{C}^{n}_k) / (q(n+1))$. Further, the proportion of those that have $s(C',1,0) = \hat{C}^{n}_k$ is $|\Aut{C_k}|/|G_{n}|$  Therefore, the total number $N_k$ becomes
\[
N_k = \frac{|\Aut{C_k}|}{|G_{n}|} \sum_{C \in \mathcal{S}_k} \frac{|G_{n+1}|S(C,\hat{C}^{n}_i)}{|\Aut{C}|q(n+1)} 
= (q-1)! |\Aut{C_k}| \sum_{C \in \mathcal{S}_k} \frac{S(C,\hat{C}^{n}_k)}{|\Aut{C}|}.
\]

On the other hand, the number $N_k$ can be obtained by finding the number of different codes that would be generated by the algorithm if equivalent codes were not rejected at any phase of the algorithm. Let $\mathcal{T}_k$ be the set of seeds obtained during the search starting from the code $\hat{C}^{n}_k$ that were not rejected during the isomorph rejection. For each seed $D \in \mathcal{T}_k$, let $N(D)$ be the number of different seeds equivalent to $D$ obtained during the search, and let $M(D)$ be the number of semi-canonical full codes that were obtained from the seed. Now the count becomes
\[
N_k = (q-1)! \sum_{D \in \mathcal{T}_k} N(D) M(D).
\]
Here, the factor $(q-1)!$ accounts for the permutations of $\mathbb{F}_q\setminus\{0\}$ in the first coordinate of the seed.

This check also alerts if the obtained full codes contain subsets equivalent to codes that should have been seeds but were not obtained during the search, or if any seeds that are equivalent to obtained seeds are missing.

\appendix[Perfect One-Error-Correcting $8$-ary MDS Codes]

It turns out that every nonlinear $(9,8^7,3)_8$ code $C$ has the property that there is a coordinate $i$ such that $s(C,i,v)$ is equivalent to the linear $(8,8^6,3)_8$ code for each $v$. This allows us to present the nonlinear perfect codes in terms of these shortened codes.

Let $\alpha$ be a primitive element of $\mathbb{F}_8$ with $\alpha^3+\alpha^2+1=0$. An element $x\in\mathbb{F}_8$ can be written as 
\[
x = a_2 \alpha^2 + a_1 \alpha + a_0,
\]
where $a_0,a_1,a_2 \in \{0,1\}$. We denote the element $x$ by a number in $\{0,1,\dots,7\}$ whose binary representation is $a_2a_1a_0$.

Let $C' \subseteq \mathbb{F}_8^8$ be a linear code with the generator matrix
\[
\begin{pmatrix}
1 &  &  &  &  &  & 1 & 1 \\
 & 1 &  &  &  &  & 1 & 2 \\
 &  & 1 &  &  &  & 1 & 3 \\
 &  &  & 1 &  &  & 1 & 4 \\
 &  &  &  & 1 &  & 1 & 5 \\
 &  &  &  &  & 1 & 1 & 6
\end{pmatrix}.
\]
Now a nonlinear $(9,8^7,3)_8$ code $C$ can be expressed as 
\[
C = \bigcup_{v \in \mathbb{F}_8} e(g^v C', i, v),
\] 
where $i \in \{1,2,\dots,9\}$ and $g^v = (\pi^v; \sigma^v_1, \sigma^v_2, \dots, \sigma^v_8) \in G_8$ for permutations $\pi^v$ of $\{1,2,\dots,8\}$ and permutations $\sigma^v_j$ of $\mathbb{F}_8$ as defined in Section~\ref{sec:equivalence}.

Selecting the coordinate $i$ corresponds to permuting the coordinates of the perfect code, so we may choose for example $i=1$. For each of the three nonlinear equivalence classes of $(9,8^7,3)_8$, one choice of these permutations to generate one representative is given in Table~\ref{tab:nonlinperfect}. The permutations $\pi^v$ of $\{1,2,\dots,8\}$ are expressed as $\pi^v(1)\pi^v(2)\cdots\pi^v(8)$ and the permutation $\sigma^v_i$ of $\mathbb{F}_8$ as $\sigma^v_i(0)\sigma^v_i(1)\cdots\sigma^v_i(7)$.

\begin{table*}
\centering
\caption{Permutations for Constructing the Nonlinear Perfect $(9,8^7,3)_8$ Codes}
\label{tab:nonlinperfect}
\renewcommand{\arraystretch}{1.3}
\begin{tabular}{|c|c|c|c|c|c|c|c|c|}
\hline
$v$ & 0 & 1 & 2 & 3 & 4 & 5 & 6 & 7 \\
\hline
\hline
$\pi^v$ & 12836457 & 12564378 & 12743685 & 12743685 & 12835764 & 12564378 & 12835764 & 12836457 \\
$\sigma^v_1$ & 06452371 & 05371246 & 01436752 & 04653127 & 06715432 & 07426153 & 02476351 & 03624157 \\
$\sigma^v_2$ & 03624157 & 02471536 & 07536142 & 05643721 & 06235417 & 07356124 & 07426153 & 01346725 \\
$\sigma^v_3$ & 07163542 & 20635174 & 04652731 & 06235417 & 03142657 & 60541732 & 07613524 & 05731264 \\
$\sigma^v_4$ & 01346725 & 05217346 & 03625741 & 06542317 & 06457123 & 01436752 & 03146275 & 70615324 \\
$\sigma^v_5$ & 07623154 & 25734610 & 04326157 & 03652471 & 04617523 & 61327540 & 03564271 & 01736425 \\
$\sigma^v_6$ & 02764351 & 01374625 & 40236157 & 30714265 & 02347615 & 07265413 & 05632174 & 03247165 \\
$\sigma^v_7$ & 03624157 & 02476351 & 25347610 & 32175460 & 40723165 & 07531624 & 60235471 & 01346725 \\
$\sigma^v_8$ & 06523147 & 01243675 & 02375641 & 03547216 & 67531240 & 04765312 & 13742560 & 42753610 \\
\hline
\hline
$\pi^v$ & 12653478 & 12736854 & 12738645 & 12734586 & 12843576 & 12654387 & 12845367 & 12735468 \\
$\sigma^v_1$ & 02314675 & 03746512 & 06247153 & 02641375 & 05634127 & 07253416 & 04375162 & 01543267 \\
$\sigma^v_2$ & 04157362 & 06421753 & 07435216 & 01457632 & 04372615 & 03561247 & 05726314 & 03476521 \\
$\sigma^v_3$ & 07526134 & 05367241 & 01672345 & 03216547 & 02751364 & 50413627 & 70165234 & 02135746 \\
$\sigma^v_4$ & 01765423 & 04152637 & 02754631 & 07523461 & 01546732 & 01674523 & 01453276 & 40263751 \\
$\sigma^v_5$ & 05172463 & 03567214 & 07462351 & 70236451 & 20617354 & 07321546 & 05176324 & 05361742 \\
$\sigma^v_6$ & 02453716 & 06352174 & 10627534 & 01752634 & 05367124 & 25673140 & 05247136 & 02541673 \\
$\sigma^v_7$ & 04253617 & 27345160 & 25176340 & 23567410 & 07264153 & 01327456 & 05736124 & 24731650 \\
$\sigma^v_8$ & 01765423 & 40516273 & 05361427 & 06174253 & 34256170 & 02315764 & 75643120 & 04627315 \\
\hline
\hline
$\pi^v$ & 12438765 & 12438765 & 12347658 & 12347658 & 12347658 & 12438765 & 12438765 & 12347658 \\
$\sigma^v_1$ & 02164753 & 04635172 & 02471653 & 07563124 & 01735462 & 06573421 & 03712546 & 05326741 \\
$\sigma^v_2$ & 06573421 & 03712546 & 04273651 & 07561342 & 03715264 & 07241365 & 01456237 & 05146723 \\
$\sigma^v_3$ & 03712546 & 02164753 & 32671450 & 45736210 & 56423170 & 01456237 & 06573421 & 74165320 \\
$\sigma^v_4$ & 04635172 & 41672350 & 02537461 & 03641725 & 07316542 & 64521730 & 35746120 & 06124357 \\
$\sigma^v_5$ & 01273645 & 03754261 & 06527431 & 02341765 & 07213546 & 07465312 & 05621437 & 03164257 \\
$\sigma^v_6$ & 05261473 & 01637245 & 30165724 & 10547263 & 70216435 & 06743152 & 03425716 & 40321576 \\
$\sigma^v_7$ & 01456237 & 60754312 & 04631527 & 03257146 & 01372654 & 50236741 & 70423156 & 02745361 \\
$\sigma^v_8$ & 07652431 & 02513647 & 02746351 & 04531627 & 06415732 & 05341276 & 01467325 & 05123476 \\
\hline
\end{tabular}
\end{table*}


\bibliographystyle{IEEEtran}
\bibliography{IEEEabrv,paper}

\begin{IEEEbiographynophoto}{Janne I. Kokkala}
was born in Espoo, Finland, in 1988. 
He received the B.Sc.\ (Tech.) degree in engineering physics from Helsinki University of Technology (TKK), Espoo, Finland, in 2009, and 
the M.Sc.\ (Tech.) degree in engineering physics from Aalto University School of Science, Espoo, Finland, in 2013.

He is currently a Doctoral Candidate with the Department of Communications and Networking at Aalto University School of Electrical Engineering, working towards the D.Sc.\ (Tech.) degree in information theory.
\end{IEEEbiographynophoto}

\begin{IEEEbiographynophoto}{Denis S. Krotov}
Denis S. Krotov was born in Novosibirsk, Russia, in 1974. He received
the Bachelor's degree in mathematics in 1995 and the Master's degree in
1997, both from Novosibirsk State University, the Ph.D. and Dr.Sc. degrees
in Discrete Mathematics and Theoretical Cybernetics from Sobolev Institute
of Mathematics, Novosibirsk, in 2000 and 2011, respectively.

Since 1997, he has been with Theoretical Cybernetics Department, Sobolev
Institute of Mathematics, where he is currently a Leading Researcher. In
2003, he was a Visiting Researcher with Pohang University of Science and
Technology, Korea. His research interest includes subjects related to discrete
mathematics, algebraic combinatorics, coding theory, and graph theory.
\end{IEEEbiographynophoto}

\begin{IEEEbiographynophoto}{Patric R. J. \"Osterg{\aa}rd}
was born in Vaasa, Finland, in 1965. He received
the M.Sc.\ (Tech.) degree in electrical engineering and the D.Sc.\ (Tech.)
degree in computer science and engineering, in 1990 and 1993,
respectively, both from Helsinki University of Technology TKK, Espoo,
Finland.

From 1989 to 2001, he was with the Department of Computer Science and
Engineering at TKK. During 1995--1996, he visited Eindhoven University of
Technology, The Netherlands, and in
2010 he visited Universit\"at Bayreuth, Germany. Since 2000,
he has been a Professor at TKK---which merged with two other universities
into the Aalto University in January 2010---currently in the Department
of Communications and Networking. He was the Head of the Communications
Laboratory, TKK, during 2006--2007. He is the coauthor of
\emph{Classification Algorithms for Codes and Designs} (Berlin: Springer-Verlag, 2006), and,
since 2006, co-Editor-in-Chief of the \emph{Journal
of Combinatorial Designs}. His research interests include algorithms,
coding theory, combinatorics, design theory, information theory, and
optimization.

Dr.\ \"Osterg{\aa}rd is a Fellow of the Institute of Combinatorics and its
Applications. He is a recipient of the 1996 Kirkman Medal and was
awarded an honorary doctorate, Doctor et Professor Honoris Causa, by the
University of P\'ecs, Hungary, in 2013.
\end{IEEEbiographynophoto}

\end{document}